\documentclass[10pt,journal,final]{IEEEtran}
\usepackage{graphicx,amssymb,amsmath,amsthm}
\usepackage{enumerate}
\usepackage{comment,cite,color}
\usepackage{cite,color}
\usepackage{mathrsfs}
\usepackage{epsfig}
\usepackage{lscape}
\usepackage{subfigure}
\usepackage{epstopdf}
\usepackage{caption}
\usepackage{array}
\usepackage{algorithm}
\usepackage{algpseudocode}

\begin{document}
\renewcommand{\algorithmicrequire}{\textbf{Input:}}
\renewcommand{\algorithmicensure}{\textbf{Output:}}
\newcommand{\PP}{\mathbb{P}}
\newcommand{\HH}{\mathbf{H}}
\newcommand{\innerp}[1]{\langle {#1} \rangle}
\newcommand{\norm}[1]{\|{#1}\|_2}
\newcommand{\abs}[1]{\lvert#1\rvert}
\newcommand{\absinn}[1]{\vert\langle {#1} \rangle\rvert}
\newcommand{\argmin}[1]{\mathop{\rm argmin}\limits_{#1}}
\newcommand{\argmax}[1]{\mathop{\rm argmax}\limits_{#1}}
\newcommand{\ds}{\displaystyle}
\newcommand{\wt}{\widetilde}
\newcommand{\ra}{{\rightarrow}}
\newcommand{\lra}{{\longrightarrow}}
\newcommand{\eproof}{\hfill\rule{2.2mm}{3.0mm}}
\newcommand{\esubproof}{\hfill$\Box$}
\newcommand{\Proof}{\noindent {\bf Proof.~~}}
\newcommand{\D}{{\mathcal D}}
\newcommand{\TT}{{\mathcal T}}
\renewcommand{\SS}{{\mathcal S}}
\newcommand{\BS}{{\mathbb S}}
\newcommand{\Prob}{{\rm Prob}}
\newcommand{\R}{{\mathbb R}}
\newcommand{\Z}{{\mathbb Z}}
\newcommand{\T}{{\mathbb T}}
\newcommand{\C}{{\mathbb C}}
\newcommand{\CG}{{\mathcal G}}
\newcommand{\Q}{{\mathbb Q}}
\newcommand{\N}{{\mathbb N}}
\newcommand{\ep}{\varepsilon}
\newcommand{\wmod}[1]{\mbox{~(mod~$#1$)}}
\renewcommand{\eqref}[1]{(\ref{#1})}
\newcommand{\inner}[1]{\langle #1 \rangle}
\newcommand{\shsp}{\hspace{1em}}
\newcommand{\mhsp}{\hspace{2em}}
\newcommand{\FT}[1]{\widehat{#1}}
\newcommand{\conj}[1]{\overline{#1}}
\newcommand{\ber}{\nu_{\lambda}}
\newcommand{\bern}{\nu_{\lambda_n}}
\newcommand{\biasber}{\nu_{\lambda, p}}
\newcommand{\bequiv}{\sim_\lambda}
\newcommand{\bnequiv}{\sim_{\lambda_n}}

\newcommand{\E}{{\mathbb E}}
\newcommand{\B}{{\mathcal B}}
\newcommand{\MC}{{\mathcal C}}
\newcommand{\ML}{{\mathcal L}}
\newcommand{\MLP}{{\mathcal{PL}}}
\newcommand{\MLS}{{\mathcal{PS}}}
\newcommand{\NN}{{\mathcal N}}
\renewcommand{\i}{{\mathbf i}}
\newcommand{\rank}{{\rm rank}}
\newcommand{\supp}{{\rm supp}}
\newcommand{\diag}{{\rm diag}}
\renewcommand{\j}{{\mathbf j}}
\newcommand{\vx}{{\mathbf x}}
\newcommand{\vy}{{\mathbf y}}
\newcommand{\vO}{{\mathbf 0}}
\newcommand{\vo}{{\mathbf o}}
\newcommand{\va}{{\mathbf a}}
\newcommand{\vb}{{\mathbf b}}
\newcommand{\vd}{{\mathbf d}}
\newcommand{\vv}{{\mathbf v}}
\newcommand{\vu}{{\mathbf u}}
\newcommand{\vf}{{\mathbf f}}
\newcommand{\vg}{{\mathbf g}}
\newcommand{\ve}{{\mathbf e}}
\renewcommand{\H}{{\mathbb H}}
\newcommand{\qHH}{\widetilde{\HH}}
\newcommand{\qC}{\widetilde{\C}}
\newcommand{\ul}[1]{\underline{#1}}
\newcommand{\MM}{\mathbf M}
\newcommand{\XX}{\mathbf X}
\newcommand{\x}{\mathbf x}
\newcommand{\ba}{\mathbf a}
\newcommand{\F}{{\mathcal F}}
\newcommand{\G}{{\mathcal G}}
\newcommand{\CR}{{\mathcal R}}
\newcommand{\sn}{{\mathbb S}^n}
\newcommand{\vn}{{\mathbf v}^n}
\newcommand{\un}{{\mathbf u}^n}
\newcommand{\uD}{\underline{D}}
\newcommand{\dd}{{\mathrm d}}
\newcommand{\s}{{\mathbf s}}
\newcommand{\zz}{^{\top}}
\renewcommand{\theequation}{\thesection.\arabic{equation}}
\renewcommand{\thefigure}{\arabic{figure}}
\newtheorem{definition}{Definition}[section]
\newtheorem{corollary}{Corollary}[section]
\newtheorem{theorem}{Theorem}[section]
\newtheorem{example}{Example}[section]
\newtheorem{lemma}{Lemma}[section]
\newtheorem{remark}{Remark}[section]
\newtheorem{notation}{Notation}[section]
\newtheorem{prop}{Proposition}[section]

\title{ Phaseless recovery using Gauss-Newton method \thanks{Research of Zhiqiang Xu was supported  by NSFC grant ( 11422113,  91630203, 11331012) and by National Basic Research Program of China (973 Program 2015CB856000).}\thanks{ B.~Gao and Z.~Xu are from Inst. Comp. Math., Academy of Mathematics and Systems Science,
Chinese Academy of Sciences, Beijing, China. Email: gaobing@lsec.cc.ac.cn, xuzq@lsec.cc.ac.cn} }
\author{ Bing Gao,\,\, Zhiqiang Xu}
\date{\today}
\markboth{Date of current version June~2017}%
{Shell \MakeLowercase{\textit{et al.}}: Bare Demo of IEEEtran.cls for IEEE Journals}
\maketitle

\begin{abstract}
In this paper, we propose a Gauss-Newton algorithm to recover a $ n $-dimensional signal from its phaseless measurements. Our algorithm has two stages: in the first stage, our algorithm gets a good initialization by calculating the eigenvector corresponding to the largest eigenvalue of a Hermitian matrix; in the second stage, our algorithm solves an optimization problem iteratively using the Gauss-Newton method. Our initialization method makes full use of all measurements and provides a good initial guess as long as the number of random measurements is  $O(n)$. For real-valued signals, we prove that a re-sampled version of Gauss-Newton iterations can converge to the global optimal solution quadratically with $ O(n\log n) $ random measurements. Numerical experiments show that Gauss-Newton method has better empirical performance over the other algorithms, such as Wirtinger flow algorithm and  alternating minimization algorithm, etc.
\end{abstract}
\begin{IEEEkeywords}
phaseless recovery, phase retrieval,  Gauss-Newton method, quadratic convergence.
\end{IEEEkeywords}

\section{Introduction}

\subsection{Phaseless Recovery Problem}
\IEEEPARstart{R}{ecovering} a signal from the magnitude of measurements, known as {\em phaseless recovery problem}, frequently occurred in science and engineering \cite{application1,application2,application3,application4}.
Suppose that $\{a_1,\ldots,a_m\}\subset \H^n$ is a frame, i.e., ${\rm span}\{a_1,\ldots,a_m\}=\H^n$ $(\H=\C$ or $\R)$ and $ y_j =\abs{\innerp{a_j,z}}^2,\, j=1,\ldots,m$ where $  z\in \H^n $.
The phaseless recovery problem can be formulated in the form of solving quadratic equations:
\begin{equation}\label{eq:phase}
y_j=\abs{\innerp{a_j,x}}^2, j=1,\ldots,m,
\end{equation}
where $a_j\in \H^n$ are the sensing vectors.
Our aim is to recover $ z $ (up to a global unimodular constant) by solving (\ref{eq:phase}).

Recently phaseless recovery problem attracts much attention \cite{BCE06, BCMN, bodmann} and many algorithms are developed for solving it.
A well-known method is the error reduction algorithm \cite{alt1,alt2}.
Despite the algorithm is used in many applications, there are few theoretical guarantees for the global convergence of it.
In \cite{PN13},  a re-sampled version of the error reduction  algorithm, the Altmin Phase algorithm, is introduced with proving that the algorithm geometrically converges to the true signal up to an accuracy of $\epsilon$ provided the measurement matrix $A:=[a_1,\ldots,a_m]^*\in \C^{m\times n}$ is Gaussian random matrix with $m=O(n\log^3n\log\frac{1}{\epsilon})$. In fact, to attain the accuracy of $\epsilon$, the algorithm needs $O(\log\frac{1}{\epsilon})$ iterations and different measurements are employed in each iteration of the algorithm. Wirtinger flow (WF) method was first introduced to solve the phaseless recovery problem in \cite{WF}. WF method combines a good initial guess, which is obtained by spectral method, and a series of updates that refine the initial estimate by  a deformation of the gradient descent method. It is proved that WF method converges to an exact solution on a linear rate from $O(n\log n)$ Gaussian random measurements  \cite{WF}. In fact, it is shown in \cite{WF}  that
\[
{\rm dist}(x_{k+1}-z)\,\,\leq\,\, \rho\cdot {\rm dist} (x_k-z),
\]
where $x_k$ is the output of the $k$-th iteration of WF method, $z$ is the true signal, $0<\rho<1$ is a constant and the definition of ${\rm dist}(\cdot)$ is given in Section \ref{notations}.  The truncated WF method is introduced in  \cite{TWF}, which improves the performance of WF method with showing that $O(n)$ Gaussian random measurements are enough to attain  the linear convergence rate. Recently another two stage iterative algorithm, the truncated amplitude flow (TAF), was proposed by Wang, Giannakis and Eldar in \cite{null_initial2}. The TAF uses null initialization method to obtain an initial estimate and refines it by successive updates of truncated generalized gradient iterations. It is  proved that TAF can geometrically converge to the exact signal with $O(n)$ measurements \cite{null_initial2}.
Despite iterative algorithms to solve phaseless recovery problem, a recent approach   is to recast phaseless recovery as a semi-definite programming (SDP), such as PhaseLift \cite{lift1,lift2,lift3}. PhaseLift is to lift a vector problem to a rank-1 matrix one and then one can recover the rank-1 matrix by minimizing the trace of matrices.  Though PhaseLift can provide the exact solution using $O(n)$ measurements,  the computational cost is large when the dimension of the signal is high.

In many applications, the signals to be reconstructed are known to be sparse in advance, i.e., most of the elements are equal to zero. Thus it is natural to develop algorithms to recover sparse signals from the magnitude of measurements, which is also known as {\em sparse phaseless recovery} or {\em sparse phase retrieval}. The $\ell_1$ model for the recovery of sparse signals from the magnitude of measurements is studied in \cite{sparse1,sparse2,sparse3}.  A greedy algorithm, GESPAR, for solving sparse phase retrieval is presented in \cite{gespar}. The core step of the method is to  use the damped Gauss-Newton method to solve a non-linear least square problem. They choose the step size by backtracking and prove that damped Gauss-Newton method converges  to a stationary point. In \cite{TWF} and \cite{Caisp}, the authors investigate the performance of modified WF method for the recovery of real-valued sparse signals from phaseless measurements.

\subsection{Our contribution}
The aim of this paper is twofold. We first present  an alternative    initial guess
which is the eigenvector corresponding to the largest eigenvalue of
\begin{small}
\[
\frac{1}{m}\sum_{j=1}^{m}\left(\frac{1}{2}-\exp\left(\frac{-y_j}{\sum_r y_r/m}\right)\right)a_ja_j^*.
\]
\end{small}
Compared with the one obtained by the spectral method \cite{WF}, the new initial guess can reach accuracy with $O(n)$ Gaussian random measurements while the spectral method requires $O(n\log n)$. The numerical experiments also show that our initialization method has a better performance over the other previous methods.
Our second aim  is to set up a new iterative algorithm for solving phaseless recovery problem. In the algorithm, starting with our initial guess,  we  refine the initial estimation by iteratively applying an update rule, which comes from a Gauss-Newton iteration.  Thus for the convenience of description, we name this algorithm as Gauss-Newton algorithm.
Under the assumption of $z\in \R^n$  and $A\in \C^{m\times n}$ being  Gaussian random matrix, we investigate the performance of the Gauss-Newton algorithm with showing that a re-sampled version of it can quadratically converge to the true signal up to a global sign, i.e.,
\[
{\rm dist}(x_{k+1},z)\leq \beta\cdot  ({\rm dist}(x_k,z))^2,
\]
where $x_k$ is the output of the $k$-th iteration and $\beta$ is a constant. Hence, 
to reach the accuracy  $\epsilon$, re-sampled Gauss-Newton method needs $ O(\log\log \frac{1}{\epsilon})$ iterations, which has an  improvement over the Altmin Phase algorithm. Since many signals from real world are real, the assumption of $z$ being real is reasonable. For the case where signals are complex, we derive a revised Gauss-Newton method.

\subsection{Notations}\label{notations}
Throughout the paper, we reserve $ C $, $ c $ and $ \gamma $, and their indexed versions to denote positive constants. Their value vary with the context. We use $ z \in \H^n $ to denote the target signal.
When no subscript is used,  $ \|\cdot\| $ denotes the Euclidian norm, i.e., $ \|\cdot\|= \|\cdot\|_2$.
We use the Gaussian random vectors $ a_j\in\H^n, \, j=1,\ldots, m $ as the sampling vectors and obtain $ y_j = |\langle a_j, z\rangle|^2,\, j=1,\ldots,m $. Here we say the sampling vectors are the Gaussian random measurements if
$ a_j\in\C^n, \, j=1,\ldots, m $ are i.i.d. $ \mathcal{N}(0,I/2) +i \mathcal{N}(0,I/2) $ random variables or $ a_j\in\R^n, \, j=1,\ldots, m $ are i.i.d. $ \mathcal{N}(0,I) $ random variables.  Denote $ x_k $ as the output of the $ k $-th iteration and $ S_k $ as the line segment between $ x_k $ and $ z $, i.e.,
\[
S_k\,\,:=\,\, \{t z+(1-t)x_k: 0\leq t\leq 1\}.
\]
As the problem setup naturally leads to ambiguous solutions,  we define \[
\{ cz\in \C^n: |c|=1\}
\]
as the solution set.
Then we define
$$
\text{dist}(x,z)=\left\{
\begin{aligned}
& \min_{\phi\in[0,2\pi)}\|z-e^{i\phi}x\|  &\quad \H=\C, \\
& \min\{\|z-x\|,\|z+x\|\}  &\quad \H=\R.
\end{aligned}
\right.
$$
as the distance between $x\in \H^n$ and the solution set.
\subsection{Organization}
The rest of this paper is organized as follows. In Section \ref{sec_initial}, we introduce a new initialization method and prove that it can provide a good initial guess by only $O(n)$ Gaussian random measurements. The Gauss-Newton algorithm for phaseless recovery problem is discussed in Section \ref{sec_gn}. Under the assumption of signals being real and the measurement matrix $A\in \C^{m\times n}$ being complex Gaussian random matrix,  we prove that a re-sampled version of this algorithm can achieve quadratic convergence. Some numerical experiments are given in Section \ref{sec_experiments} to illustrate the practical efficiency of the Gauss-Newton algorithm. At last, most of the detailed proofs are given in the Appendix.

\section{Initialization}\label{sec_initial}

\subsection{Initialization method}
For non-convex problem (\ref{eq:phase}),  proper initial criteria is essential to avoid the iterative algorithm trapping in a local minimum.
So the first step of Gauss-Newton method is to choose an initial estimation. Before giving our initialization method, we first  review several other  methods.

Spectral initialization method \cite{alt1,alt2,WF} estimates the initial guess $ z_0 $  as the eigenvector corresponding to the largest eigenvalue of $\frac{1}{m}\sum_{j=1}^{m}y_j a_ja_j^* $
with norm \begin{small}
	$  \sqrt{\sum_{j=1}^{m}y_j/m} $.
\end{small} In \cite{WF}, Cand\`{e}s, Li and Soltanolkotabi prove that when $a_j, \,j=1,\ldots,m$ are Gaussian random measurements with $m\geq C_0n\log n$,
$ {\rm dist}(z_0,z)\leq 1/8\|z\| $
holds with probability at least $1-10\exp(-\gamma n)-8/n^2$. To reduce the number of observations, a modified spectral method is introduced in \cite{TWF}, which precludes $ y_j $ with large magnitudes. Particularly, they select the initial value as the eigenvector corresponding to the largest eigenvalue  of $ \frac{1}{m}\sum_{j=1}^{m}y_j a_ja_j^*I_{\{|y_j|\leq \beta_y\lambda^2\}} $,
where $ \beta_y $ is an appropriate truncation criteria and \begin{small}
	$ \lambda^2=\sum_jy_j/m $.
\end{small} This method only requires the number of measurements $m\geq Cn$ with a sufficient large constant $ C $.
The null initialization method is introduced by Chen, Fannjiang and Liu in \cite{null_initial1}.
This method builds on the orthogonality characteristics of high-dimensional random vectors, and
choose the eigenvector corresponding to the largest eigenvalue of $ \frac{1}{|I|}\sum_{j\in I} a_j a_j^* $
as the initial guess, where $ I $ is an index set selected by the $ |I| $ largest magnitudes of $ \frac{|\langle a_j, z\rangle|^2}{\|a_j\|^2\|z\|^2} $, $ j=1,\ldots,m $. When the number of measurements is on the order of $ n $, null initialization method can guarantee a good precision.
More details can be found in \cite{null_initial1}, \cite{null_initial2}.
To state conveniently,  we name the first method as SI (Spectral Initialization), the second method as TSI (Truncated Spectral Initialization) and the third method as NI (Null Initialization).

Next we introduce a new  method for initialization, which is stated in Algorithm \ref{initialization1}.
In fact, the initial guess is chosen as the eigenvector corresponding to the largest eigenvalue of the  Hermitian matrix
\begin{small}
\[
Y:=\frac{1}{m}\sum_{j=1}^{m}\left(\frac{1}{2}-\exp\big(-\frac{y_j}{\lambda^2}\big)\right)a_ja_j^*
\]
\end{small}
and normalized by $\lambda:= \sqrt{\frac{1}{m}\sum_{j=1}^{m}y_j} $. We next briefly introduce the reason why we choose the matrix $Y$.
When $ a_j\in\C^n $, $ j=1,\ldots,m $ are the Gaussian random measurements, we have
\begin{small}
	\[
	\E(Y_1) = \frac{zz^*}{4\|z\|^2},
	\]
\end{small}
where
\begin{small}
	\[
	Y_1:=\frac{1}{m}\sum_{j=1}^{m}\left(\frac{1}{2}-\exp\big(-\frac{y_j}{\|z\|^2}\big)\right)a_ja_j^*.
	\]
\end{small}
Noting that $ \lambda^2 $ is a good approximation to $ \|z\|^2 $. So we choose $ Y $ as an approximation to $ \dfrac{zz^*}{4\|z\|^2} $, whose eigenvector of the largest eigenvalue is of the form $ cz $ where  $ c $ is a constant.
Meanwhile, the eigenvector associated with the largest eigenvalue  of $ Y $ can be efficiently calculated by the power method (see details in \cite{WF}).

The new method makes full use of every observation and can obtain an alternative  initial value by nearly optimal number of measurements (see Theorem \ref{initial1_1}). Beyond theoretical results, numerical experiments also show that this method has better performance than that of SI, TSI and NI (see Example \ref{initial_exam}).
\begin{algorithm}[H]
	\caption{Initialization}\label{initialization1}
	\begin{algorithmic}[H]
		\Require
		Observations $ {y}\in\mathbb{R}^m $.\\
		Set
		$$
		\lambda^2=\frac{\sum_jy_j}{m}.
		$$
		Set $ x_0 $, normalized to $ \|x_0\|_2=\lambda $, to be the eigenvector corresponding to the largest eigenvalue of
		$$
		 Y=\frac{1}{m}\sum_{j=1}^{m}\left(\frac{1}{2}-\exp\big(-y_j/\lambda^2\big)\right)a_ja_j^*.
		$$
		\Ensure
		Initial guess $ x_0 $.
	\end{algorithmic}
\end{algorithm}

\subsection{The performance of Algorithm \ref{initialization1}}
The following theorem provides theoretical analysis of Algorithm \ref{initialization1}. Here we suppose $ z\in\C^n $, $a_j\in\C^n$ and prove that the initial guess $ x_0 $ is not far from $ cz, |c|=1 $.
\begin{theorem}\label{initial1_1}
	Suppose that $a_j\in \C^n, j=1,\ldots,m$ are Gaussian random measurements, $z\in \C^n$ and $x_0$ is the output of
	Algorithm \ref{initialization1}.  For any $ \theta>0 $, there exists a constant $C_\theta$ such that when $ m\geq C_\theta n$,
	\begin{equation}\label{eq:inidis}
	\text{\rm dist}(x_0, z)\,\,\leq\,\, \sqrt{3\theta}\|z\|
	\end{equation}
	holds with probability at least $ 1-4\exp(-c_\theta n) $, where $ c_\theta>0 $.
\end{theorem}

\begin{remark}\label{real_initial_notation}
	Theorem \ref{initial1_1} only considers the case where both $z$ and $a_j$ are complex vectors. For the case where $z$ is real and $a_j$ is complex, we can choose the initial guess $x_0$ as the vector corresponding to the largest eigenvalue of the real part of the matrix $Y$. The (\ref{eq:inidis}) still holds for this case.
	When $ z\in\R^n $, $ a_j\in\R^n, \,j=1,\ldots,m $,  we set
	\begin{small}
		\[
		Y_R:=\frac{1}{m}\sum\limits_{j=1}^{m}\left(\frac{1}{\sqrt{3}}-\exp(-|a_j\zz z|^2/\lambda^2)\right)a_ja_j\zz,
		\]
	\end{small}
	where $ \lambda^2 = \dfrac{1}{m}\sum\limits_{j=1}^{m}y_j $. Then for any $ \theta>0 $, we have
	\begin{small}
			$$
		\left\|Y_R -\frac{2\sqrt{3}}{9\|z\|^2} zz\zz \right\|\leq\frac{\sqrt{3}}{9}\theta
		$$
	\end{small}
	with probability at least $ 1-4\exp(-c_\theta n) $ provided $ m\geq C_\theta n $. Using similar method with the proof of Theorem \ref{initial1_1}, we can obtain
	$$
	\textup{dist}(x_0,z)=\min\{\|x_0-z\|,\|x_0+z\| \}\leq \sqrt{3\theta}\|z\|.
	$$
\end{remark}

\begin{remark}
	It is possible to obtain similar results with replacing $\exp(-y_j/\lambda^2)$ in $Y$ by another bounded function $g(y_j)$. For example, we can take $g(y_j)=\exp(-y_j^p/\lambda^2)$ where $0<p\leq 1$.
 We need adjust the constant $1/2$ in $Y$ when we replace the function $\exp(-y_j/\lambda^2)$ in $Y$ by another bounded function $g(y_j)$.
\end{remark}

\section{ Gauss-Newton Method }\label{sec_gn}

In this section, we present Gauss-Newton iterations which are used to refine the initial guess.

\subsection{Real-valued signals}
We first consider the  case where the exact signal $ z $ is real-valued.  For every measurement vector $ a_j\in\C^n $, we use $ a_{jR} $ and $ a_{jI} $ to represent its real and imaginary part respectively, i.e., $a_j = a_{jR}+ia_{jI} $,
$  a_{jR}\in\R^n, \, a_{jI}\in\R^n,\,  j=1,\ldots,m $. Then we can rewrite  (\ref{eq:phase}) as a nonlinear least square problem
\begin{eqnarray}\label{eq:mreal}
\min\limits_{x\in\R^n}  f(x):=\frac{1}{2m}\sum\limits_{j=1}^m \left(
\langle a_{jR}, x\rangle^2+\langle a_{jI}, x\rangle^2  - y_j
\right)^2,
\end{eqnarray}
where $ y_j = \abs{\langle a_j, z\rangle}^2 $.
To state conveniently,  we set $F_j(x) := \frac{1}{\sqrt{m}}(\langle a_{jR}, x\rangle^2+\langle a_{jI}, x\rangle^2   - y_j)$ and we write (\ref{eq:mreal}) in the form of
\begin{eqnarray}\label{eq:realob}
\min\limits_{x\in\R^n} f(x)=\frac{1}{2}\sum_{j=1}^{m}F_j(x)^2.
\end{eqnarray}
\subsubsection{Gauss-Newton iteration}

To solve the nonlinear least square problem (\ref{eq:realob}), our algorithm uses the well-known Gauss-Newton iteration.
To make the paper self-contained, we introduce the Gauss-Newton iteration in detail (see also \cite{GNiter1, GNiter2}).
Suppose the $ k $-th iteration point $x_k$ is real-valued, we first linearize the nonlinear term $ F_j(x) $ at the point $x_k$:
\begin{align*}
F_j(x) &\approx F_j(x_k)+\nabla F_j(x_k)\zz (x-x_k)\\
&=\frac{1}{\sqrt{m}}\big(\langle a_{jR}, x_k\rangle^2+\langle a_{jI}, x_k\rangle^2-y_j\\
&\quad \quad +2(a_{jR}a_{jR}\zz x_k+a_{jI}a_{jI}\zz x_k)\zz(x-x_k)\big).
\end{align*}
Suppose that the $j$-th row of $J(x_k)\in \R^{m\times n}$  is $\frac{2}{\sqrt{m}}(a_{jR}a_{jR}\zz x_k\\+a_{jI}a_{jI}\zz x_k)\zz  $ and the $j$-th component of $F(x_k)\in \R^m$ is given by $F_j(x_k),\,\, j=1,\ldots,m$.
Then the following least square problem can be considered as an approximation to (\ref{eq:realob}):
\begin{equation}\label{leastsq}
\min_{{x}\in\R^n} \quad \frac{1}{2}\|J(x_k){(x-x_k)}+F(x_k)\|_2^2.
\end{equation}
We choose the next iteration point $ x_{k+1} $ as the solution to (\ref{leastsq}), i.e.,
\begin{align}\label{update_rule}
x_{k+1}& = x_k-\big(J(x_k)\zz J(x_k)\big)^{-1}J(x_k)\zz F(x_k)\\\nonumber
& = x_k-\big(J(x_k)\zz J(x_k)\big)^{-1}\nabla f(x_k),\nonumber
\end{align}
where
\begin{small}
\begin{align}\label{expressionJJ}
\nonumber J(x_k)\zz J(x_k)&=\frac{4}{m}\sum_{j=1}^{m}\bigg((a_{jR}\zz x_k)^2a_{jR}a_{jR}\zz +(a_{jI}\zz x_k)^2a_{jI}a_{jI}\zz\\
&\quad \quad +(a_{jR}\zz x_k)(a_{jI}\zz x_k)(a_{jI}a_{jR}\zz+a_{jR}a_{jI}\zz)\bigg)
\end{align}
\end{small}
and
\begin{small}
\begin{align}\label{expressionJF}
\nonumber& J(x_k)\zz F(x_k) = \nabla f(x_k)\\
 &=\frac{2}{m}\sum_{j=1}^{m}\left((a_{jR}\zz x_k)^2+(a_{jI}\zz x_k)^2-y_j\right)(a_{jR}a_{jR}\zz x_k+a_{jI}a_{jI}\zz x_k).
\end{align}
\end{small}
Thus we obtain the update rule (\ref{update_rule}). Note that the $ x_{k+1} $ is also real-valued.
\subsubsection{Gauss-Newton Method with Re-sampling}
The Gauss-Newton method uses Algorithm \ref{initialization1} to obtain an initial guess $ x_0 $ and iteratively refine $ x_k $ by the update rule (\ref{update_rule}). In theoretical analysis, as we require that the current measurements are independent with the last iteration point (see Ramark \ref{resampling}), we re-sample measurement matrix $ A  $ in every iteration step. Then Algorithm \ref{GN} is in fact a variant of Gauss-Newton method with using different measurements in each iteration.  The re-sampling idea is also used in \cite{PN13} for the alternating minimization algorithm  and in \cite{WF} for the WF  algorithm with coded diffraction patterns.

\begin{algorithm}[!h]
\caption{Gauss-Newton Method with Re-sampling}\label{GN}
	\begin{algorithmic}[h]
		\Require
		Measurement matrix: $ A\in\C^{m\times n}$, observations: $  y\in \R^m $ and
		$ \epsilon>0 $.
		\begin{enumerate}
			\item[1:] Set $ T=c\log\log\frac{1}{\epsilon} $, where $ c $ is a sufficient large constant.
			\item[2:] Partition $ y $ and the corresponding rows of $ A $ into $ T+1 $ equal disjoint sets: $ (y^{(0)}, A^{(0)})$, $(y^{(1)}, A^{(1)})$, $\ldots, (y^{(T)}, A^{(T)})$. The number of rows in $A^{(j)}$ is $m'=m/(T+1)$.
			\item[3:] Set
			$
			\lambda: = \sqrt{1/m'\sum_{j}y_j^{(0)}}.
			$
			Set $ x_0 $ to be the eigenvector corresponding to the largest eigenvalue of the real part of
			\begin{small}
				$$
				 \frac{1}{m}\sum_{j=1}^{m}\left(1/2-\exp(-y^{(0)}_j/\lambda^2)\right)a_j^{(0)}a_j^{(0)}{^*}
				$$
			\end{small}
			with $ \|x_0\|=\lambda $.
			\item[4:] For $ k=0, 1,\ldots, T-1 $ do
			\begin{small}
			   \begin{align*}
				x_{k+1}
				&=x_k - \\
				&\big(J^{k+1}{(x_k)}\zz J^{k+1}(x_k)\big)^{-1} J^{k+1}(x_k)\zz F^{k+1}(x_k).
				\end{align*}
			\end{small}
			\item[5:] End for
		\end{enumerate}
		\Ensure
		$ x_T $.
	\end{algorithmic}
\end{algorithm}

In step 4 of the Algorithm \ref{GN}, the concrete form of matrix \begin{small}
	$J^{k+1}{(x_k)}\zz J^{k+1}(x_k)$
\end{small} and vector
\begin{small}
$J^{k+1}(x_k)\zz F^{k+1}(x_k)$
\end{small} is same with (\ref{expressionJJ}) and (\ref{expressionJF}).
Here $y_1,\ldots,y_{m'}$ are the entries of   $y^{(k+1)}$  and $a_1,\ldots,a_{m'}$
are the rows of $A^{(k+1)}$.
\subsubsection{Convergence Property of Gauss-Newton Method with Re-sampling}
We next present  theoretical convergence property of  Algorithm \ref{GN}.  Without loss of generality, we assume $ \|z\|=1 $.
Theorem \ref{th:maintheorem} illustrates that under given conditions, Algorithm \ref{GN} has a quadratic convergence rate. Furthermore, we show that to achieve an $ \epsilon $ accuracy, the Gauss-Newton method with re-sampling only needs $ O(\log\log(\frac{1}{\epsilon})) $ iterations.
\begin{theorem}\label{th:maintheorem}
	Let $z\in \R^n$ with $ \|z\|=1 $ be an arbitrary vector and $0<\delta \leq 1/93$ be a constant.
	Suppose that $x_k\in \R^n$ satisfies  $ \text{dist}(x_k, z)\leq\sqrt{\delta} $.
	Suppose
	$y_j=\abs{\innerp{a_j,z}}^2$, where $a_j\in \C^n$, $ j=1,\ldots, m $ are Gaussian random measurements with $m\geq Cn\log n$.
	The $x_{k+1}$ is defined by the update rule (\ref{update_rule}). Then with probability at least
	$ 1-c/n^2 $, we have
	\begin{equation}\label{distance_rela}
	\textup{dist}(x_{k+1}, z)\leq\beta\cdot \textup{dist}^2(x_{k}, z),
	\end{equation}
	where
	\begin{small}
	\begin{equation}\label{eq:beta}
	\beta = \frac{8(7+\frac{3}{4}\delta)(1+\sqrt{\delta})}{(8-\delta)(1-\sqrt{\delta})^2}\leq\frac{1}{\sqrt{\delta}}.
	\end{equation}
	\end{small}	
\end{theorem}
\begin{remark}\label{lessthanone}
	In Theorem \ref{th:maintheorem}, the reason why we require $0<\delta\leq 1/93$ is
	to guarantee $\beta\cdot\delta \leq \sqrt{\delta}$. Hence the condition $\textup{dist}(x_{k+1}, z)\leq\beta\cdot \delta\leq\sqrt{\delta}$  still holds and we can  use Theorem
	\ref{th:maintheorem} at the $(k+1)$-th iteration.
\end{remark}
According to Theorem \ref{initial1_1} or Remark \ref{real_initial_notation}, for any $ 0<\delta\leq 1/93 $ and $ 0< \theta\leq\delta/3 $, when $m\geq C_\theta n$, it holds  with probability at least $ 1-4\exp(-c_\theta n) $ that
\[
\text{dist}(x_0,z)\leq \sqrt{3\theta}\leq\sqrt{\delta}.
\]
Combining this initialization result with Theorem \ref{th:maintheorem}, we have the following conclusion.
\begin{corollary}\label{iter_step}
	Suppose that $ z\in\R^n $ with $ \|z\|=1 $ is an arbitrary vector and $ a_j\in\C^n $, $ j=1,\ldots,m $ are Gaussian random measurements. Suppose that $\epsilon$ is an arbitrary constant within range $(0,1/2)$ and $\delta\in (0,1/93]$ is a fixed constant.  If $ m\geq C \cdot \log\log\frac{1}{\epsilon}\cdot n\log n $, then with probability at least $ 1-\tilde{c}/n^2 $, Algorithm \ref{GN} outputs $ x_{T} $ such that
	\[
	\textup{dist}(x_{T},z)\,\,<\,\, \epsilon,
	\]
	where $C$ is a constant depending on $\delta$, $ \epsilon $.
\end{corollary}
\begin{IEEEproof}
	According to Theorem \ref{initial1_1} or Remark \ref{real_initial_notation}, we have
	\[
	{\rm dist}(x_0,z)\,\,\leq \,\, \sqrt{\delta}
	\]
	with probability at least $1-4\exp(-c_\delta n)$.
	From Remark \ref{lessthanone}, we know
	\[
	\beta\cdot\delta\leq \sqrt{\delta},
	\]
	where $\beta$ is defined in Theorem \ref{th:maintheorem}. In Algorithm \ref{GN}, we choose $ T = c\log\log\frac{1}{\epsilon} $ and $ m'\geq C_1n\log n $, where $ C_1 $ is a constant depending on $ C,\,c $.
	Iterating (\ref{distance_rela}) in Theorem \ref{th:maintheorem} $ T $ times leads to
	\begin{align*}
	\text{dist}(x_{T},z)&\leq\beta\cdot \text{dist}^2(x_{T-1}, z)\\
	&\leq \beta^{2^T-1}\text{dist}^{2^{T}}(x_0, z)\\
	&\leq \beta^{2^T-1}\cdot(\sqrt{\delta})^{2^{T}}\\
	&\leq (\beta\cdot\sqrt{\delta})^{2^{T}}\\
	&\leq\epsilon,
	\end{align*}
	which holds with probability at least  $ 1-\tilde{c}/n^2 $.
\end{IEEEproof}

\begin{remark}\label{resampling}
In Algorithm 2, we use different measurement vectors in each iteration.
In fact, Theorem III.1 requires that the Gaussian random measurement vectors $a_j$
are independent with $x_k$.
According to (\ref{update_rule}), $x_{k+1}$ depends on the current measurement vectors $a_j$.
Hence, to use Theorem III.1 at the next step, we need choose different measurement vectors which are independent with the previous ones.
\end{remark}
\subsection{Complex-valued signals}
In this subsection, we consider the case where the signal $z$ is complex.
To recover $z\in \C^n$ from $y_1,\ldots,y_m$, we need to solve the following programming:
\begin{equation}\label{c_leastsquare}
\min\limits_{x\in\C^n}  \tilde{f}(x):=\frac{1}{2m}\sum\limits_{j=1}^m (
\abs{ a_{j}^* x}^2 - y_j
)^2:=\frac{1}{2m}\sum\limits_{j=1}^{m}\tilde{F}_j(x)^2,
\end{equation}
where $ \tilde{F}_j(x)= x^* a_ja_{j}^* x - y_j $.
To state conveniently, for $x\in \C^n$,
set
\begin{equation}\label{eq:jing}
x^\sharp := \begin{pmatrix}
x-x_{k}\\ \overline{x}-\overline{x}_{k}
\end{pmatrix}\in\C^{2n}.
\end{equation}
Using a similar argument with above, at the $k$-th iteration, we can update $x_k$ by solving
\begin{equation}\label{cleast}
\min_{x\in\C^n}  \|A_k x^\sharp + \tilde{F}_k \|_2^2,
\end{equation}
where
\[
\tilde{F}_k := \tilde{F}(x_{k}) := (|a_1^*x_{k}|^2-y_1,\ldots, |a_m^*x_{k}|^2-y_m)\zz,
\]
\[
A_k := (\tilde{J}(x_{k}), \overline{\tilde{J}(x_{k})})= \begin{pmatrix}
x_{k}^* a_1a_1^*, & x_{k}\zz\overline{a}_1a_1\zz\\
\vdots & \vdots\\
x_{k}^* a_ma_m^*, & x_{k}\zz\overline{a}_m a_m\zz\\
\end{pmatrix}\in\C^{m\times 2n}.
\]

In fact, if ${\hat x}\in \C^n$ is a solution to (\ref{cleast}) then we can
update $x_k$ by $x_{k+1}={\hat x}+x_k$.
However, the following proposition shows that the solution to (\ref{cleast}) is not unique.
\begin{prop}\label{pr:inv}
	Suppose that ${\hat x}$ is a special solution to (\ref{cleast}). Then for any $c_0\in \R$,
	${\hat x}+ic_0x_k$ is also a solution to (\ref{cleast}) where $i=\sqrt{-1}$.
\end{prop}
\begin{IEEEproof}
	Noting that  $
	A_k\begin{pmatrix}
	x_{k}\\-\overline{x_{k}}
	\end{pmatrix}=0,
	$ we have
	\begin{eqnarray*}
		A_k({\hat x}+ic_0x_k)^\sharp&=A_k\begin{pmatrix}
			\hat{x}+ic_0x_k-x_{k}\\ \overline{\hat x}-ic_0\overline{x}_{k}-\overline{x}_k
		\end{pmatrix}\\
		&=A_k\begin{pmatrix}
			\hat{x}-x_{k}\\ \overline{\hat x}-\overline{x}_k
		\end{pmatrix}=A_k\hat{x}^\sharp,
	\end{eqnarray*}
	which implies the conclusion. Here, we  use the definition of $x^\sharp$ (see (\ref{eq:jing})).
\end{IEEEproof}

%

We denote the solution set to (\ref{cleast}) as ${\mathcal L}_k$.	Our idea is to choose $ \hat{x}\in {\mathcal L}_k $ so that $ \|{x}_{k+1}-x_{k}\|_2 =\|\hat{x}\|_2$ reaches the minimum since we already know $ x_k $ is not far from the exact signal.
 Then we have
 \begin{prop}\label{pr:guangyiinv}
 We use  $ A_k^{\dag} $ to denote the moore-penrose pseudoinverse of $ A_k $. Then
\[
-A_k^{\dag}\tilde{F}_k(1:n)=\argmin{x\in {\mathcal L}_k}\|x\|_2,
\]
 where
$ A_k^{\dag}\tilde{F}_k(1:n) $  denotes the vector consisting of the first $ n $ elements of $ A_k^{\dag}\tilde{F}_k $.
 \end{prop}
 \begin{proof}
 According to the property of the moore-penrose pseudoinverse of $ A_k $ \cite{pseudoinverse},
  $-A_k^{\dag}\tilde{F}_k  $ is the minimal norm least square solution to
  \begin{equation}\label{eq:guangzui}
\min_{u\in\C^{2n}}  \|A_k u + \tilde{F}_k \|_2^2.
\end{equation}
We claim that $A_k^{\dag}\tilde{F}_k$ satisfies
\begin{equation}\label{eq:gonge}
A_k^{\dag}\tilde{F}_k(1:n)\,\,=\,\, \overline{A_k^{\dag}\tilde{F}_k(n+1:2n)}
\end{equation}
where $ A_k^{\dag}\tilde{F}_k(1:n) $ and $A_k^{\dag}\tilde{F}_k(n+1:2n)$ denote the vectors consisting of the first and the last $ n $ elements of $ A_k^{\dag}\tilde{F}_k $, respectively.
 And hence $ -A_k^{\dag}\tilde{F}_k(1:n) $ is  a solution to (\ref{cleast}), i.e., $-A_k^{\dag}\tilde{F}_k(1:n)\in {\mathcal L}_k$. Since $-A_k^{\dag}\tilde{F}_k  $ is the minimal
norm least square solution to  (\ref{eq:guangzui}), we obtain the conclusion. We still need show (\ref{eq:gonge}).
 Recall that
\[
A_k^\dag=\lim_{\delta\rightarrow 0}A_k^*\left(A_kA_k^*+\delta I\right)^{-1},
\]
which implies (\ref{eq:gonge}) since $A_kA_k^*+\delta I$ is  a real matrix.
 \end{proof}
 Then we can take the iteration step as
	\begin{equation}\label{eq:upd}
	x_{k+1} = -A_k^{\dag}\tilde{F}_k(1:n) + x_k.
	\end{equation}
    The numerical experiments show (\ref{eq:upd}) has quadratic convergence rate provided the initial guess is not far from the exact signal (see Example \ref{convergence_exam} (b)). The analysis of the convergence property of (\ref{eq:upd}) is the subject of our future work.


\section{Numerical Experiments }\label{sec_experiments}
The purpose of  numerical experiments  is to compare the performance of Gauss-Newton method with that of other existing methods as mentioned above.
In our numerical experiments,  the measurement matrix $A\in \C^{m\times n}$ is generated by Gaussian random measurements and the entries of the original signal $z\in \H^n$ is drawn from standard normal distribution.
\begin{example}\label{initial_exam}
	In this example, we test Algorithm 1 to compare the initial guess of Algorithm 1 with that of spectral initialization (SI), truncated spectral initialization (TSI) and null initialization (NI). For $\H=\C$, we take $n=128$ and change $m$ within the range $[4n, 22n]$. For each $m$, 50 iterations of power method are run to calculate the eigenvectors. We repeat the experiment $50$ times and record the average value of the relative error $ \textup{dist}(x_0, z)/\|z\| $.  Figure 1
	depicts that Algorithm 1 outperforms SI,TSI and NI significantly.
\end{example}
\begin{figure}[!htb]
	\begin{center}
{
\includegraphics[width=0.43\textwidth]{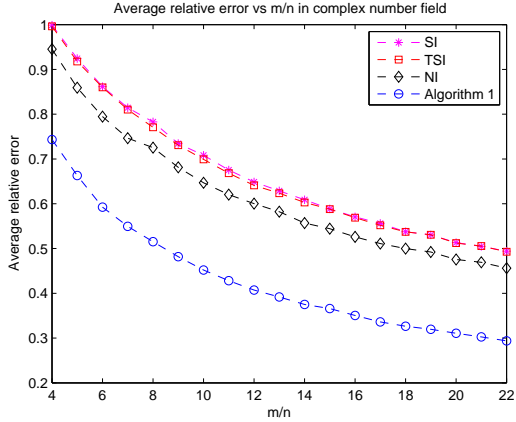}}
		\caption{Initialization experiments: Averaged relative error between $ x_0 $ and $ z $ for $ n = 128 $ and $ m/n $  changing within the range $ [4,22] $. The figures show  that Algorithm 1 performs better than the others in terms of average relative error.}
	\end{center}
\end{figure}

\begin{example}\label{convergence_exam}
In this example, we compare the convergence rate of Gauss-Newton method with that of WF method \cite{WF}, of Altmin Phase method \cite{PN13} and of TAF method \cite{null_initial2}. We take $n=128$, $m/n=5$. Here we use noisy Gaussian data model $ y_j = |\langle a_j, z\rangle|^2+\eta_j ,\,j=1,\ldots,m$, where $ \eta_j $ is chosen according to $ \eta_j\sim\mathcal{N}(0, 0.1^2) $.  We choose the original signal $z\in \R^n$ for (a) and $z\in \C^n$ for (b). When $z\in \C^n$, we use iteration (\ref{eq:upd}) to update the iteration point. Figure 2 depicts the relative error against the iteration number. The numerical results show that Gauss-Newton method  has the better performance in the noisy measurements and converges faster  over the other methods.
\end{example}
\begin{figure*}[!htb]
\begin{center}
\subfigure[]{
\includegraphics[width=0.45\textwidth]{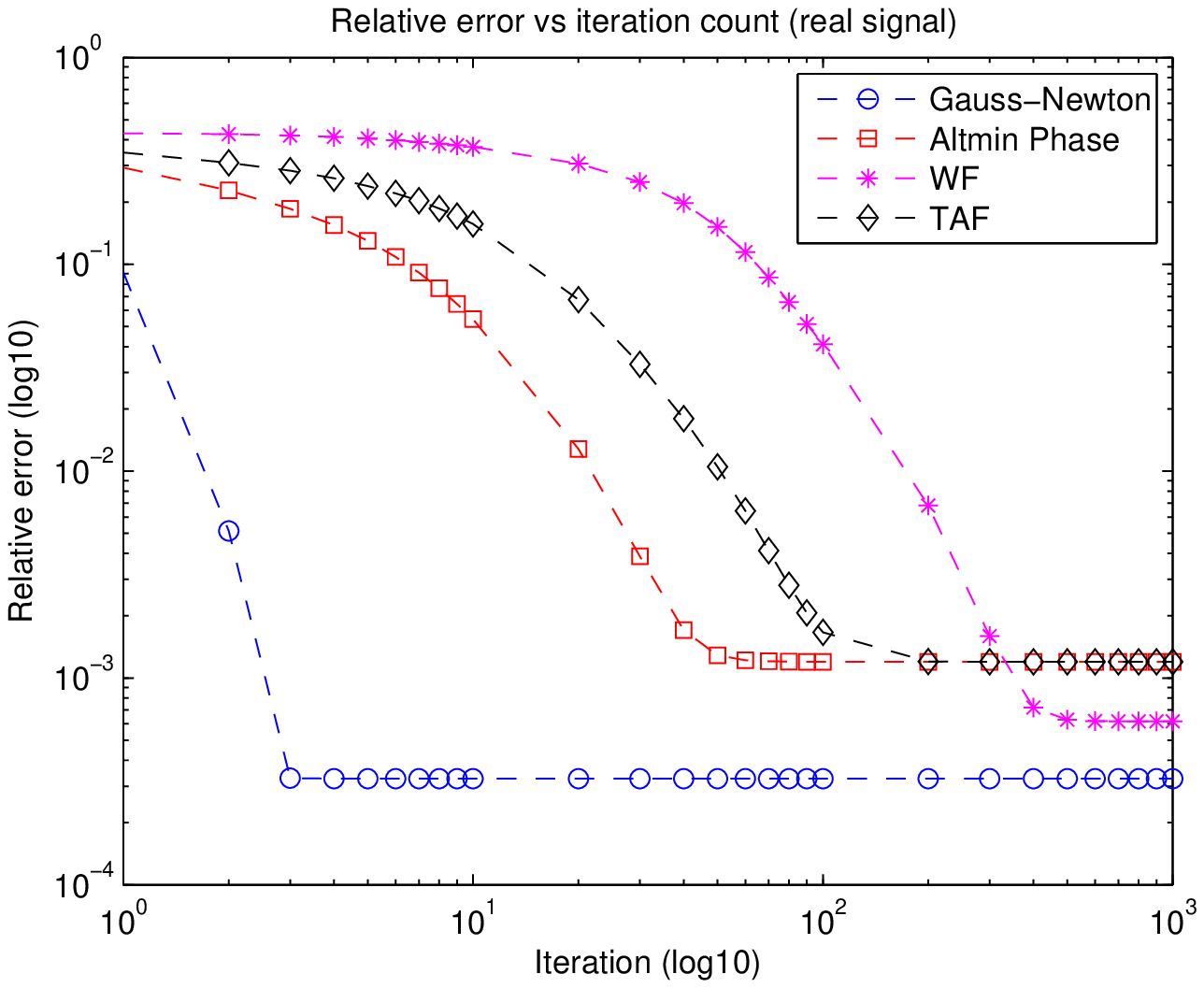}}
\subfigure[]{
\includegraphics[width=0.45\textwidth]{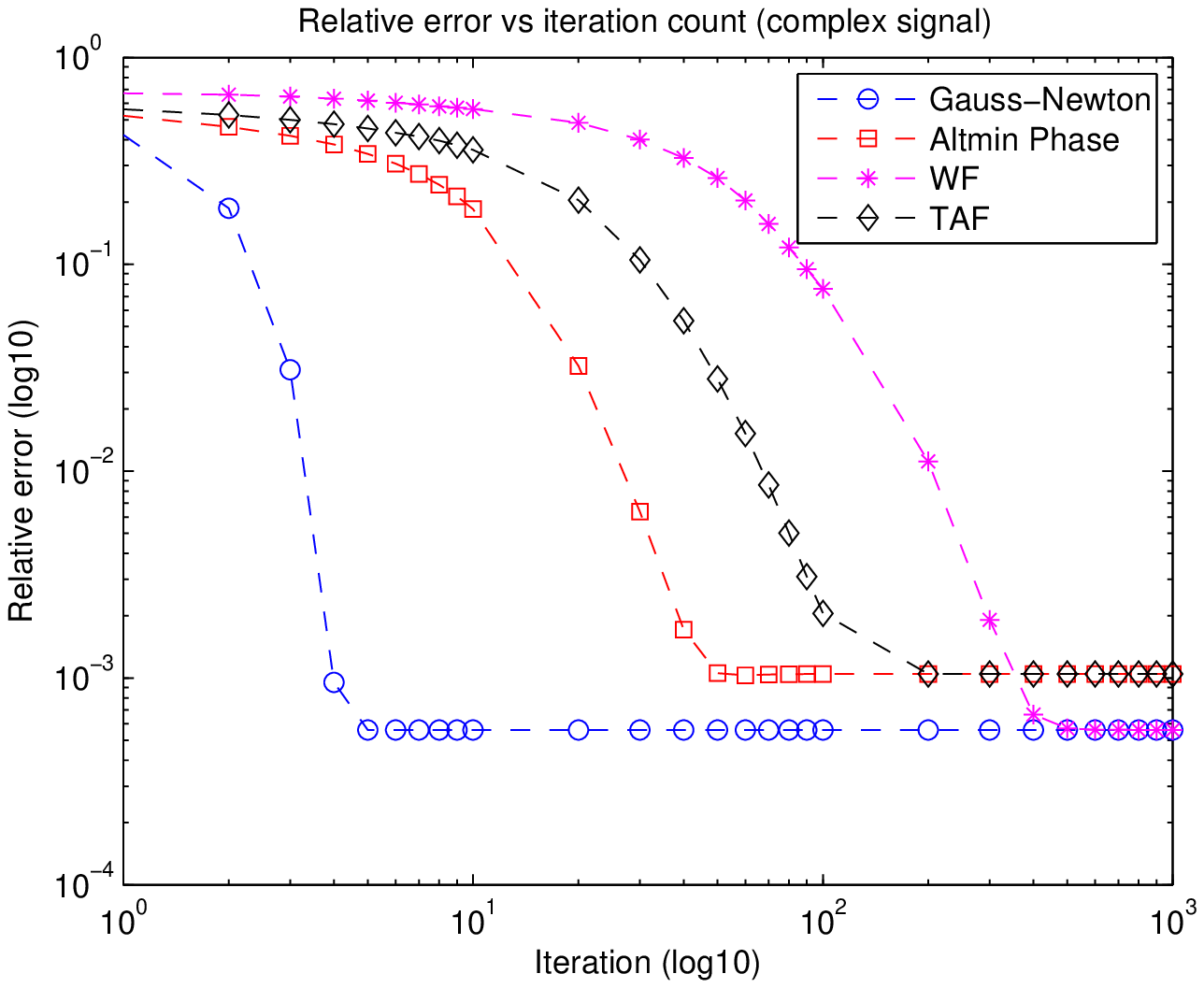}}
\caption{Convergence experiments: Plot of relative error ($ \log(10) $) vs number of iterations ($ \log(10) $) for Gauss-Newton method, Altmin Phase method, WF method and TAF method. Take $ n=128 $, $ m=5n $. The measurements are corrupted with Gaussian noise. The figure (a) (for real signal) and figure (b) (for complex signal) both show that Gauss-Newton method provides better solution and converges faster. }
\end{center}
\end{figure*}

\begin{example}
We compare the CPU time taken by different methods.  For all of them, we use Algorithm \ref{initialization1} to generate the same initial value. That means we only compare the CPU time consumed by the iteration step. Here we define the CPU time of this trial to be the time used until the first iteration after which the relative error is smaller than $ 10^{-5} $. We take $ n=128 $, $ m=5n $, $ z\in\R^n $ and $ y_j = |\langle a_j, z\rangle|^2 ,\,j=1,\ldots,m $. Table \ref{cputime} records the CPU time of these methods and shows that Gauss-Newton method is the most time-saving method.	
\end{example}
\renewcommand\arraystretch{2}
\begin{table}[!htb]		
\begin{center}
\begin{tabular}{|c|c|c|c|c|}
\hline
& Gauss-Newton & Altermin Phase & WF & TAF  \\\hline
Iter & 4  & 68 & 522& 189\\\hline
CPU(s) & 0.0313 & 1.2500 & 1.1719 & 0.0938\\\hline
\end{tabular}
\end{center}
	\caption{CPU Time}
	\label{cputime}
\end{table}
\begin{example}
	In this example, we test the success rate of Gauss-Newton method. Let $ z\in\R^n $ and take $n=128$ and change $m/n$ within the range $[1,10]$ with the step size $0.5$.  For each $m/n$, we repeat 100 times and calculate the success rate. Here we claim a trial successful when the relative error is smaller than $ 10^{-5} $. Figure 3 shows the numerical results with using the recovery algorithm Gauss-Newton, WF, Altmin Phase and TAF, respectively. The figures show that Gauss-Newton method and TAF method can achieve a success rate of $ 100\% $ when $ m/n\geq 3 $, which is much better than WF and Altmin Phase.
\end{example}

\begin{figure}[!htb]
	\begin{center}
		{
		\includegraphics[width=0.45\textwidth]{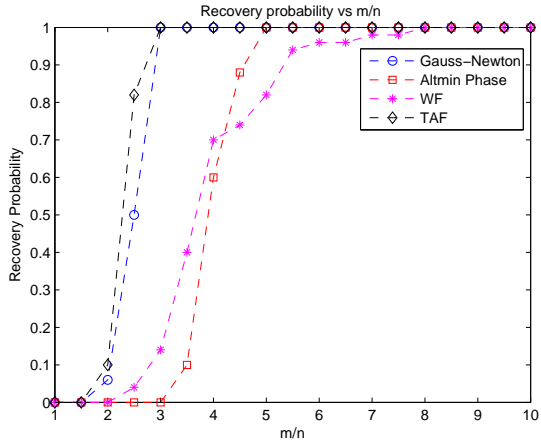}}
		\caption{Success rate experiments: Empirical probability of successful recovery based on 100 random trails for different $ m/n $. Take $ n=128 $ and change $ m/n $ between 1 and 10.  The figures demonstrate that Gauss-Newton method and TAF method are better than WF method and Altmin Phase method in terms of success rate.}
	\end{center}
\end{figure}

\section{Appendix}
\subsection{Proof of Theorem \ref{initial1_1}}
To prove the Theorem \ref{initial1_1}, we first recall some useful results.
\begin{theorem}[Wely Theorem]
	Suppose $ A $, $ B \in \C^{n\times n}$ are two Hermitian matrices. The eigenvalues of $ A $ are denoted as $ \lambda_1\geq\lambda_2\geq\ldots\geq\lambda_n $ and the eigenvalues of $ B $ are denoted as $ \mu_1\geq\mu_2\geq\ldots\geq\mu_n $. Then we have
	\[
	|\mu_i-\lambda_i|\leq\|A-B\|_2, \quad i=1,2,\ldots,n.
	\]
\end{theorem}
\begin{lemma}[Theorem 5.39 in \cite{non-asym}]\label{chernoff-Hoeffding}
	Assume that $ a_j\in\C^n $, $ j=1,\ldots,m $ are independent sub-gaussian random vectors. For any $ \zeta>0 $, when the number of samples obeys $ m\geq C_\zeta\cdot n $,
	\begin{equation}\label{gauss_exp}
	 \left\|\frac{1}{m}\sum\limits_{j=1}^{m}a_ja_j^*-\E\bigg(\frac{1}{m}\sum\limits_{j=1}^{m}a_ja_j^*\bigg) \right\|\leq\zeta
	\end{equation}
	holds with probability at least $ 1-2\exp(-c_\zeta m) $. Here $ C_\zeta $, $ c_\zeta $ depend on the constant $ \zeta $ and the sub-gaussian norm $ \max_j\|a_j\|_{\psi_2} $.
\end{lemma}
The next lemma plays an essential role in proving  Theorem \ref{initial1_1}.
\begin{lemma}\label{mainlemma1_complex}
	Let $ z\in \C^n $ be a fixed vector. Suppose $ a_j\in\C^n,\,\, j=1,2,\ldots,m $ are the Gaussian random measurements and $y_j=\abs{\innerp{a_j,z}}^2, \,\, j=1,\ldots,m$. Set
	\[
	\lambda^2 = \frac{1}{m}\sum_{j=1}^{m}y_j
	\]
	and
	\begin{small}
	$$
	 Y\,\,:=\,\,\frac{1}{m}\sum\limits_{j=1}^{m}\left(\frac{1}{2}-\exp\Big(-\frac{|a_j^* z|^2}{\lambda^2}\Big) \right)a_ja_j^*.
	$$
	\end{small}
	Then for any $ 0<\eta<1 $,
	\[
	\left\|Y-\frac{zz^*}{4\|z\|^2}\right\|\leq\eta
	\]
	holds with probability at least $ 1-4\exp(-c_\eta n)$ provided  $ m\geq C_\eta n $, where $ c_\eta>0 $, $ C_\eta $ are constants depending on $ \eta $.
\end{lemma}
\begin{IEEEproof}
	Set
	\begin{small}
	\[
	Y_1\,\, :=\frac{1}{m}\sum\limits_{j=1}^{m}\left(\frac{1}{2}-\exp\left(-\frac{|a_j^* z|^2}{\|z\|^2}\right) \right)a_ja_j^*.
	\]
	\end{small}	
	As $ a_j\in\C^n,\,\, j=1,2,\ldots,m $ are the Gaussian random measurements, a simple moment calculation gives
	$$
	\E (Y_1)\,\,=\,\, \frac{zz^*}{4\|z\|^2}.
	$$
	Then
	\begin{align}\label{exp}
	\left\|Y-\frac{zz^*}{4\|z\|^2}\right\|=\|Y-\E (Y_1)\|\leq\|Y - Y_1\|+\|Y_1-\E (Y_1)\|.
	\end{align}
	We first consider  the second term of (\ref{exp}), i.e., $\|Y_1-\E (Y_1)\|$.
	Note that both $ a_j $ and $ \sqrt{\exp(-|a_j^*z|^2/\|z\|^2)}a_j $, $ j=1,\ldots,m $ are all sub-gaussian random vectors with
	\begin{small}
		\begin{equation*}
		\E \left(\frac{1}{m}\sum\limits_{j=1}^{m}a_ja_j^*\right)=I_n
		\end{equation*}
	\end{small}
	and
	\begin{small}
		\begin{equation*}
		\E \left(\frac{1}{m}\sum\limits_{j=1}^{m}\exp\left(-\frac{|a_j^*z|^2}{\|z\|^2}\right)a_ja_j^*\right)=\frac{1}{2}I_n-\frac{zz^*}{4\|z\|^2}.
		\end{equation*}
	\end{small}	
	Using Lemma \ref{chernoff-Hoeffding}, 	we obtain that
	\begin{small}
	\begin{equation}\label{exp_21}
	 \bigg\|\frac{1}{m}\sum\limits_{j=1}^{m}a_ja_j^*-I_n\bigg\|\leq\frac{1}{4}\eta
	\end{equation}
	\end{small}	
	and
	\begin{small}
		\begin{equation}\label{exp_22}
		 \bigg\|\frac{1}{m}\sum\limits_{j=1}^{m}\exp\left(-\frac{|a_j^*z|^2}{\|z\|^2}\right)a_ja_j^*-\left(\frac{1}{2}I_n-\frac{zz^*}{4\|z\|^2}\right)\bigg\|\leq\frac{1}{8}\eta
		\end{equation}
	\end{small}	
	holds with probability at least $ 1-4\exp(-c_\eta n) $ provided $ m\geq C_\eta n $, where $ c_\eta$, $ C_\eta $ are constants depending on $ \eta $ and subgassian norm of  $ a_j $, $ \sqrt{\exp(-|a_j^*z|^2/\|z\|^2)}a_j $, $ j=1,\ldots,m $. The inequality (\ref{exp_21}) also implies that
	\begin{equation}\label{lambda}
	\big|\lambda^2 - \|z\|^2\big|\leq\frac{\eta}{4}\|z\|^2
	\end{equation}
	holds with high probability.
	Combining  (\ref{exp_21}) and (\ref{exp_22}), we obtain that
\begin{small}
\begin{align}\label{exp_2}
&\|Y_1-\E Y_1\|\\\nonumber
&=\bigg\|\frac{1}{m}\sum\limits_{j=1}^{m}\frac{1}{2}a_ja_j^*-\frac{1}{m}\sum\limits_{j=1}^{m}\exp\left(-\frac{|a_j^*z|^2}{\|z\|^2}\right) a_ja_j^*\\\nonumber
&\quad\quad\quad-\frac{1}{2}I_n+\frac{1}{2}I_n-\frac{zz^*}{4\|z\|^2}\bigg\|\\\nonumber
&\leq\frac{1}{2}\bigg\|\frac{1}{m}\sum\limits_{j=1}^{m}a_ja_j^*-I_n\bigg\|\\\nonumber
&\quad\quad\quad+\bigg\|\frac{1}{m}\sum\limits_{j=1}^{m}\exp\left(-\frac{|a_j^*z|^2}{\|z\|^2}\right)a_ja_j^*-\left(\frac{1}{2}I_n-\frac{zz^*}{4\|z\|^2}\right)\bigg\|\\\nonumber
&\leq\frac{1}{4}\eta.
\end{align}
\end{small}

	For the first term of (\ref{exp}), we have
	\begin{small}
		\begin{align*}
			&Y- Y_1\\
			&= \frac{1}{m}\sum\limits_{j=1}^{m}\left(\exp\Big(-\frac{|a_j^* z|^2}{\|z\|^2}\Big)-\exp\Big(-\frac{|a_j^* z|^2}{\lambda^2}\Big) \right)a_ja_j^*.
		\end{align*}
	\end{small}
	By (\ref{exp_21}), we have
\begin{small}
	\begin{align}\label{exp_1}
	&	\|Y-Y_1 \|\\\nonumber
	&\leq \max_j\left|\exp\bigg(-\frac{|a_j^* z|^2}{\|z\|^2}\bigg)-\exp\bigg(-\frac{|a_j^* z|^2}{\lambda^2}\bigg)\right|\cdot\left(1+\frac{1}{4}\eta\right)\\\nonumber &=\max_j\exp\bigg(-\frac{|a_j^*z|^2}{\xi}\bigg)\frac{|a_j^*z|^2}{\xi^2}\cdot\big|\|z\|^2-\lambda^2\big|\cdot\left(1+\frac{1}{4}\eta\right)\\\nonumber
		&\leq\frac{1}{\xi}\cdot\frac{\eta}{4}\|z\|^2\cdot(1+\frac{1}{4}\eta)\leq\frac{3}{4}\eta,
		\end{align}
\end{small}
where the second inequality dues to (\ref{lambda}) and the fact that $ x\exp(-x)<1 $ for all $ x $. The
	second line of (\ref{exp_1}) uses  the Lagrange's mean value theorem with  $ \xi\in [(1-\frac{\eta}{4})\|z\|^2,\,\,(1+\frac{\eta}{4})\|z\|^2] $ with high probability. Thus putting (\ref{exp_2}) and (\ref{exp_1}) into (\ref{exp}), we get
	\[
	\|Y - \E (Y_1) \|\leq\eta.
	\]
	So the conclusion holds with probability at least $ 1-4\exp(-c_\eta n) $ provided $ m\geq C_\eta n $, where $ c_\eta $, $ C_\eta $ are constants depending on $\eta $.
\end{IEEEproof}
Now we begin to prove  Theorem \ref{initial1_1}.
\begin{IEEEproof}[\bf Proof of Theorem \ref{initial1_1}]
	Suppose $ \tilde{x}_0 $ with $\|\tilde{x}_0\|=1  $ is the eigenvector corresponding to the largest eigenvalue $ \lambda_{\max}(Y) $ of
	$$
	 Y=\frac{1}{m}\sum_{j=1}^{m}\left(\frac{1}{2}-\exp(-|a_j^* z|^2/\lambda^2)\right)a_ja_j^*.
	$$
	From Lemma \ref{mainlemma1_complex}, for any $ 0< \theta\leq 1 $ and $ m\geq C_\theta n$, we have
	\begin{small}
		\[
		 \bigg\|Y-\frac{zz^*}{4\|z\|^2}\bigg\|\,\,\leq\,\,\frac{\theta}{8}
		\]
	\end{small}
	with probability at least $ 1-4\exp(-c_\theta n) $. Note that the largest eigenvalue of $ \dfrac{zz^*}{4\|z\|^2} $ is $ \dfrac{1}{4} $.
	Then according to the Wely Theorem,
	\begin{small}
	\begin{align}\label{leq_com}
	 \left|\lambda_{\max}(Y)-\frac{1}{4}\right|\,\,\leq\,\,\frac{\theta}{8}
	\end{align}
	\end{small}	
	holds with probability  at least $ 1-4\exp(-c_\theta n) $.
	On the other hand,
	\begin{small}
		\begin{align}\label{geq_com}
		 \frac{\theta}{8}&\,\,\geq\,\,\bigg\|Y-\frac{zz^*}{4\|z\|^2}\bigg\|\\\nonumber
		 &\,\,\geq\,\,\bigg|\tilde{x}_0^*(Y-\frac{zz^*}{4\|z\|^2})\tilde{x}_0\bigg|\\\nonumber
		 &\,\,=\,\,\bigg|\lambda_{\max}(Y)-\frac{1}{4}+\frac{1}{4}-\frac{|\tilde{x}_0^* z|^2}{4\|z\|^2}\bigg|\\\nonumber
		&\,\,\geq\,\,\bigg|\frac{|\tilde{x}_0^* z|^2}{4\|z\|^2}-\frac{1}{4}\bigg|-\bigg|\frac{1}{4}-\lambda_{\max}(Y)\bigg|.
		\end{align}
	\end{small}	
	Combining (\ref{leq_com}) and (\ref{geq_com}), we obtain
	\[
	\frac{|\tilde{x}_0^* z|^2}{\|z\|^2}\geq 1-\theta.
	\]
	From the proof of Lemma \ref{mainlemma1_complex} (see (\ref{lambda})), we have
	\[
	-\theta\|z\|^2\leq\lambda^2 - \|z\|^2\leq \theta\|z\|^2
	\]
	with high probability. So set $ x_0=\lambda\tilde{x}_0 $, for any $ 0< \theta\leq1 $, we have
	\begin{align*}
	\text{dist}^2(x_0,z)&=\min_{\phi\in [0,2\pi)} \|z-e^{i\phi}\lambda\tilde{x}_0\|^2\\
	&\leq  \|z\|^2+\lambda^2-2\lambda |\tilde{x}_0^* z|\\
	&\leq  \|z\|^2+(1+\theta)\|z\|^2-2{(1-\theta)}\|z\|^2\\
	&= 3\theta\|z\|^2
	\end{align*}
	with probability at least $ 1-4\exp(- c_\theta n) $ provided $ m\geq C_\theta n $. Thus we get the conclusion
	\begin{align*}
	\text{dist}(x_0, z)&\,\leq\, \sqrt{3 \theta}\|z\|.
	\end{align*}
\end{IEEEproof}

\subsection{Proof of Theorem \ref{th:maintheorem}}
In this section, we devote to prove the Theorem \ref{th:maintheorem}. At first, we give some essential lemmas.
\begin{lemma}\label{WF_lemma}[Lemma 7.4 in \cite{WF}]
	For a signal $ x\in\H^n $, suppose that $ a_j\in\C^n $, $ j=1,2,\ldots, m $ are Gaussian random measurements and $ m\geq Cn\log n $,  where  $ C $ is sufficiently large. Set
	\[
	 S:=\frac{1}{m}\sum\limits_{j=1}^{m}|a_j^*x|^2a_ja_j^*.
	\]
	Then for any $ \delta> 0 $,
	\[
	\|S-\E(S)\|\leq\frac{\delta}{4}\|x\|^2
	\]
	holds with probability at least $ 1-5\exp(-\gamma_\delta n)-4/n^2 $.
\end{lemma}
Recall that
$S_k= \{t z+(1-t)x_k: 0\leq t\leq 1\}$.
We set
\begin{small}
	\begin{align}\label{H_definition}
	 & H(x): =\nabla^2 f(x)-J(x)\zz J(x)\\\nonumber
	&=\frac{2}{m}\sum_{j=1}^{m}\left((a_{jR}\zz x)^2 +(a_{jI}\zz x)^2-y_j\right)(a_{jR}a_{jR}\zz+a_{jI}a_{jI}\zz).
	\end{align}
\end{small}
\begin{lemma}\label{J(x)}
	Suppose that  $ \|x_k - z\|\leq\sqrt{\delta} $, where
	$ x_k, z\in\R^n $ with $\|z\|=1$ and
	$ 0< \delta\leq 1/93 $ is a constant. Suppose that the measurement vectors $ a_j\in\C^n $, $ j=1,\ldots, m $ are Gaussian random measurements, which are independent with $ x_k $ and $ z $. Then when $ m\geq Cn\log n $,
	\begin{small}
		\begin{align*}
		J(x)\zz J(x)=\frac{4}{m}\sum_{j=1}^{m}\bigg(&(a_{jR}\zz x)^2a_{jR}a_{jR}\zz +(a_{jI}\zz x)^2a_{jI}a_{jI}\zz\\
		& +(a_{jR}\zz x)(a_{jI}\zz x)(a_{jI}a_{jR}\zz+a_{jR}a_{jI}\zz)\bigg)
		\end{align*}
	\end{small}
	is $ L_J $-Lipschitz continuous on $ S_k $ with probability at least $ 1-5\exp(-\gamma_\delta n)-4/n^2 $, i.e, for any $x,y \in S_k$,
	\[
	\|J(x)\zz J(x)-J(y)\zz J(y)\|\,\,\leq\,\, L_J\|x-y\|
	\]
	holds with $ L_J = 8(2+\frac{\delta}{4})(1+\sqrt{\delta}) $.
\end{lemma}
\begin{IEEEproof}
	Since the measurement vectors $ a_j $, $ j=1,\ldots,m $ are rotationally invariant and independent with $ x_k $ and $ z $, wlog, we can assume that $ z = e_1 $ and $ x_k = \|x_k\|(\alpha e_1 +\sqrt{1-\alpha^2}e_2) $, where $ \alpha = \langle x_k, z\rangle/\|x_k\| $. As $ \|x_k-z\|\leq\sqrt{\delta} $, so $ \langle x_k, z\rangle\geq 0 $, i.e., $ \alpha\geq 0 $.
	We can write $x,y\in S_k$ in the form of
	\begin{equation*}
	\begin{cases}
	x = t_1 x_k + (1-t_1)z, \,\, t_1\in[0,1],\\
	y=  t_2 x_k + (1-t_2)z,\,\, t_2\in[0,1].\\
	\end{cases}
	\end{equation*}
	For any $ x,y\in S_k $,
\begin{small}
\begin{align}\label{Lip}
&\|J(x)\zz J(x)-J(y)\zz J(y) \|\\\nonumber
&= 4\Bigg\|\frac{1}{m}\sum_{j=1}^{m}\left((a_{jR}\zz x)^2-(a_{jR}\zz y)^2\right)a_{jR}a_{jR}\zz\\\nonumber
&\quad +\frac{1}{m}\sum_{j=1}^{m}\left((a_{jI}\zz x)^2-(a_{jI}\zz y)^2\right)a_{jI}a_{jI}\zz+\frac{1}{m}\sum_{j=1}^{m}\\\nonumber
&\quad\left((a_{jR}\zz x)(a_{jI}\zz x)-(a_{jR}\zz y)(a_{jI}\zz y)\right)(a_{jR}a_{jI}\zz+a_{jI}a_{jR}\zz)\Bigg\|\\\nonumber
&=2\Bigg\|\frac{1}{m}\sum_{j=1}^{m}\begin{bmatrix}
\sigma_{R,-}I_n, &\sigma_{I,-}I_n
\end{bmatrix}\begin{bmatrix}
a_{jR}\\ a_{jI}
\end{bmatrix}\begin{bmatrix}
a_{jR}\zz, & a_{jI}\zz
\end{bmatrix}\begin{bmatrix}
\sigma_{R,+}I_n\\ \sigma_{I,+}I_n
\end{bmatrix}\\\nonumber
&+\frac{1}{m}\sum_{j=1}^{m}\begin{bmatrix}
\sigma_{R,+}I_n,& \sigma_{I,+}I_n
\end{bmatrix}\begin{bmatrix}
a_{jR}\\ a_{jI}
\end{bmatrix}\begin{bmatrix}
a_{jR}\zz, & a_{jI}\zz
\end{bmatrix}\begin{bmatrix}
\sigma_{R,-}I_n\\ \sigma_{I,-}I_n
\end{bmatrix}\Bigg\|\\\nonumber
&\leq4 \|x+y\|\|x-y\|\\\nonumber
&\quad\left\|\frac{1}{m}\sum_{j=1}^{m}\begin{bmatrix}
\kappa_1 I_n, & \kappa_2 I_n
\end{bmatrix}\begin{bmatrix}
a_{jR}\\ a_{jI}
\end{bmatrix}\begin{bmatrix}
a_{jR}\zz, & a_{jI}\zz
\end{bmatrix}\begin{bmatrix}
\kappa_1 I_n\\ \kappa_2 I_n
\end{bmatrix}\right\|,
\end{align}
\end{small}
where $ \sigma_{R,+}:=a_{jR}\zz(x+y) $, $ \sigma_{R,-}:=a_{jR}\zz(x-y) $, $ \sigma_{I,+}:=a_{jI}\zz(x+y) $, $ \sigma_{I,-}:=a_{jI}\zz(x-y) $, $ \kappa_1 := \sqrt{(a_{jR}\zz e_1)^2+(a_{jR}\zz e_2)^2} $ and $ \kappa_2 := \sqrt{(a_{jI}\zz e_1)^2+(a_{jI}\zz e_2)^2} $ and the last inequality is obtained by Cauchy-Schwarz inequality.
	Next we set
	\begin{small}
		\[
		S:=\frac{1}{m}\sum_{j=1}^{m}\begin{bmatrix}
		\kappa_1 I_n, & \kappa_2 I_n
		\end{bmatrix}\begin{bmatrix}
		a_{jR}\\ a_{jI}
		\end{bmatrix}\begin{bmatrix}
		a_{jR}\zz, & a_{jI}\zz
		\end{bmatrix}\begin{bmatrix}
		\kappa_1 I_n\\ \kappa_2 I_n
		\end{bmatrix}
		\]
	\end{small}	
    By calculation, we have $ \E (S) =  I_n + e_1e_1\zz + e_2e_2\zz  $.
	According to Lemma \ref{WF_lemma}, for $ 0<\delta\leq 1/93 $ and $ m\geq Cn\log n $,
	\[
	\|S-\E(S)\|\leq\frac{\delta}{4}
	\]
	holds with probability at least $ 1-5\exp(-\gamma_\delta n)-4/n^2 $.
	So
	\begin{equation}\label{S}
	\|S\| \leq 2+\frac{\delta}{4}.
	\end{equation}
	On the other hand,
	as $ \|x_k - z\|\leq\sqrt{\delta} $, we have
	\begin{equation}\label{magnitude_xk}
	1-\sqrt{\delta}\leq\|x_k\|\leq 1+\sqrt{\delta}.
	\end{equation}
	Thus
	\begin{align}\label{x+y}
	\quad \,\,\|x+y\|
	&=\|(\lambda_1+\lambda_2)x_k + (2-\lambda_1-\lambda_2)z\|\\\nonumber
	&\leq (\lambda_1+\lambda_2)\|x_k\|+(2-\lambda_1-\lambda_2)\\\nonumber
	&\leq 2(1+\sqrt{\delta}).
	\end{align}
	Putting (\ref{S}) and (\ref{x+y}) into (\ref{Lip}), we obtain
	\begin{small}
		\[
		\|J(x)\zz J(x)-J(y)\zz J(y)\|\leq 8(2+\frac{\delta}{4})(1+\sqrt{\delta})\|x-y\|.
		\]
	\end{small}
	So we conclude that when $ m\geq Cn\log n $, $ J(x)\zz J(x) $ is Lipschitz continuous on the line $ S_k $
	with constant $ L_J = 8(2+\frac{\delta}{4})(1+\sqrt{\delta}) $
	with probability at least $ 1-5\exp(-\gamma_\delta n)-4/n^2 $.
\end{IEEEproof}

\begin{corollary}\label{H(x)}
	Under the same conditions as in Lemma \ref{J(x)},
	\begin{small}
		\[
		H(x)=\frac{2}{m}\sum_{j=1}^{m}\left((a_{jR}\zz x)^2 +(a_{jI}\zz x)^2-y_j\right)(a_{jR}a_{jR}\zz+a_{jI}a_{jI}\zz)
		\]
	\end{small}	
	is Lipschitz continuous on $ S_k $ with Lipschitz constant
	\[
	L_H=4(1+\sqrt{\delta})(3+\frac{\delta}{4}),
	\]
	with  probability at least $ 1-5\exp(-\gamma_\delta n)-4/n^2 $.
\end{corollary}
\begin{IEEEproof}
	For any $ x,y\in S_k$, we have
	\begin{small}
	\begin{align}\label{H_ineq}
	&\|H(x)-H(y)\| \\\nonumber
	&=\bigg\| \frac{2}{m}\sum_{j=1}^{m}\left((a_{jR}\zz x)^2-(a_{jR}\zz y)^2 +(a_{jI}\zz x)^2-(a_{jI}\zz y)^2\right)\cdot\\\nonumber
	&\quad\quad\quad\quad\quad (a_{jR}a_{jR}\zz+a_{jI}a_{jI}\zz)\bigg\|\\\nonumber
	&=\bigg\|\frac{2}{m}\sum_{j=1}^{m}\left(a_{jR}\zz (x+y)\cdot a_{jR}\zz (x-y)+a_{jI}\zz (x+y)\cdot a_{jI}\zz (x-y)\right)\cdot\\\nonumber
	&\quad\quad\quad\quad\quad
	(a_{jR}a_{jR}\zz+a_{jI}a_{jI}\zz)\bigg\|\\\nonumber
	&\leq 2\|x+y\|\|x-y\|\\\nonumber
	 &\quad\quad\bigg\|\frac{1}{m}\sum_{j=1}^{m}\left((a_{jR}\zz e_1)^2+(a_{jR}\zz e_2)^2+(a_{jI}\zz e_1)^2+(a_{jI}\zz e_2)^2\right)\cdot\\\nonumber
   &\quad\quad\quad\quad\quad(a_{jR}a_{jR}\zz+a_{jI}a_{jI}\zz)
	\bigg\|.
	\end{align}
	\end{small}	
	We set
	\begin{small}
	\begin{align*}
	S:=\frac{1}{m}\sum_{j=1}^{m}\Big((a_{jR}\zz e_1)^2+&(a_{jR}\zz e_2)^2+(a_{jI}\zz e_1)^2+(a_{jI}\zz e_2)^2\Big)\cdot\\
	&(a_{jR}a_{jR}\zz+a_{jI}a_{jI}\zz).
	\end{align*}
	\end{small}	
	By calculation, we have
	$$
	\E (S)=2I_n+e_1e_1\zz+e_2e_2\zz.
	$$
	So according to Lemma \ref{WF_lemma}, for $ 0<\delta\leq 1/93 $ and $ m\geq Cn\log n $,
	\[
	\|S-\E(S)\|\leq\frac{\delta}{4}
	\]
	holds with probability at least $ 1-5\exp(-\gamma_\delta n)-4/n^2 $. So we have
	\begin{equation}\label{H_exp}
	\|S\|\leq 3+\frac{\delta}{4}.
	\end{equation}	
	Putting (\ref{H_exp}) and (\ref{x+y}) into  (\ref{H_ineq}), we have
	\[
	\|H(x)-H(y)\|\leq 4(1+\sqrt{\delta})(3+\frac{\delta}{4})\|x-y\|.
	\]
	So $ H(x) $ is Lipschitz continuous on $ S_k $ with constant $ L_H=4(1+\sqrt{\delta})(3+\frac{\delta}{4}) $.
\end{IEEEproof}
Next we present an estimation of the largest eigenvalue of $ (J(x_k)\zz J(x_k))^{-1} $.
\begin{lemma}\label{largest_eigen_jj}
	Suppose that $ \|x_k-z\|\leq\sqrt{\delta} $, where $x_k, z\in \R^n$ with $\|z\|=1$
	and $ 0<\delta\leq 1/93 $.  Suppose that $ a_j\in\C^n $, $ j=1,\ldots,m $ are Gaussian random measurements which are independent with $ x_k $. If $ m\geq Cn\log n $ for a sufficiently large constant $ C $, then with probability at least $ 1-5\exp(-\gamma_\delta n)-4/n^2 $, we have
	$J(x_k)\zz J(x_k)$ is  invertible and
	\begin{small}
	\[
	\|( J(x_k)\zz J(x_k))^{-1}\|\leq\frac{4}{(16-\delta)(1-\sqrt{\delta})^2}.
	\]
	\end{small}	
\end{lemma}
\begin{IEEEproof}
	We know
\begin{small}
		\begin{align*}
		&J(x_k)\zz J(x_k)\\
		&=\frac{4}{m}\sum_{j=1}^{m}\Big((a_{jR}\zz x_k)^2a_{jR}a_{jR}\zz +(a_{jI}\zz x_k)^2a_{jI}a_{jI}\zz\\
		&\quad \quad +(a_{jR}\zz x_k)(a_{jI}\zz x_k)(a_{jI}a_{jR}\zz+a_{jR}a_{jI}\zz)\Big)\\
		&=\frac{4}{m}\sum_{j=1}^{m}\begin{bmatrix}
		(a_{jR}\zz x_k)I_n, & (a_{jI}\zz x_k)I_n
		\end{bmatrix}\begin{bmatrix}
		a_{jR} \\ a_{jI}
		\end{bmatrix}\begin{bmatrix}
		a_{jR}\zz,& a_{jI}\zz
		\end{bmatrix}\begin{bmatrix}
		(a_{jR}\zz x_k)I_n \\ (a_{jI}\zz x_k)I_n
		\end{bmatrix}.
		\end{align*}
	\end{small}
	Set
	$$
	S:= J(x_k)\zz J(x_k).
	$$
	After a simple  calculation, we obtain
	\[
	\E (S) = 2\|x_k\|^2I_n+6x_kx_k\zz
	\]
	and the minimum eigenvalue of $ \E S $ is
	\[
	\lambda_{\min}\big(\E (S)\big)=2\|x_k\|^2.
	\]
	According to Lemma \ref{WF_lemma}, for $ 0< \delta\leq 1/93 $ and $ m\geq C n\log n $,
	\[
	\|S-\E (S)\|\leq \frac{\delta}{4}\|x_k\|^2
	\]
	holds with probability at least $ 1-5\exp(-\gamma_\delta n)-4/n^2 $.
	Then according to the Wely Theorem, we have
	\begin{align*}
	|\lambda_{\min}(S)-\lambda_{\min}\big(\E
	(S)\big)|\leq\|S- \E
	(S)\|\leq\frac{\delta}{4}\|x_k\|^2,
	\end{align*}
	which implies that
	\begin{small}
		\begin{align*}
		\lambda_{\min}(S)&\geq (2-\frac{\delta}{4})\|x_k\|^2\\
		&\geq( 2-\frac{\delta}{4})(1-\sqrt{\delta})^2.
		\end{align*}
	\end{small}	
	Here, we use (\ref{magnitude_xk}) in the last inequality.
	Then with probability at least $ 1-5\exp(-\gamma_\delta n)-4/n^2 $, we have
	\begin{small}
	\begin{align}\label{jj-1}
	\lambda_{\max} (S^{-1}) = 1/\lambda_{\min}(S)\leq\frac{4}{(8-\delta)(1-\sqrt{\delta})^2},
	\end{align}
	\end{small}
	which implies the conclusion.
\end{IEEEproof}
We next present the proof of  Theorem \ref{th:maintheorem}.
\begin{IEEEproof}[\bf Proof of Theorem \ref{th:maintheorem}]
	Without loss of generality, we suppose $ \langle x_k,z\rangle\geq 0$, i.e.,
	\[
	\text{dist}(x_k, z)=\|x_k-z\|.
	\]
	Then we just need to prove when
	$
	\|x_k-z\|\leq\sqrt{\delta}
	$
	and $  m\geq Cn\log n $,
		\begin{align*}
			\text{dist}(x_{k+1}, z)
			= \|x_{k+1}-z\|
			\leq \beta\cdot\|x_k-z\|^2
			=\beta\cdot\text{dist}^2(x_k, z)
		\end{align*}
	holds with probability at least $ 1-c/n^2 $.
	
	As $ z  $ is an exact solution to (\ref{eq:mreal}),  we have $ \nabla f(z)=H(z)=0 $.
	The definition of $x_{k+1}$ shows that
	
	\begin{align}\label{eq:1}
  x_{k+1} -  z  &= x_k -z - \big(J(x_k)\zz J(x_k)\big)^{-1} \nabla f(x_k) \\\nonumber
 &= \big(J(x_k)\zz J(x_k)\big)^{-1}\cdot \\\nonumber
 &\left[
	\big(J(x_k)\zz J(x_k)\big) \cdot (x_k - z)
	- \big(\nabla f(x_k) - \nabla f(z)\big)
	\right].\nonumber
	\end{align}
	Define $S_k:=\{x_k+t(z-x_k): 0\leq t\leq 1\}$ and $ x(t)=x_k+t(z-x_k) $.
	Then we have
		\begin{align}\label{eq:2}
		\nabla f(x_k) - \nabla f(z)  &= \nabla f\big(x(0)\big)-\nabla f\big(x(1)\big)\\\nonumber
		& = -\int_{0}^{1}\frac{\dd\big(\nabla f\big(x(t)\big)\big)}{\dd t}\dd t\\\nonumber
		&= -\int_{0}^{1}\nabla^2 f(x(t))\cdot x'(t)\dd t\\\nonumber
		& = -\frac{1}{\|x_k-z\|}\int_{S_k} \nabla^2 f(x)\cdot(z-x_k) \dd s.
		\end{align}
	The integral in (\ref{eq:2}) is interpreted as element-wise.
	Combining (\ref{H_definition}) and $H(z)=0$, we obtain
	\begin{small}
		\begin{equation}\label{eq:tuidao}
		\begin{aligned}
		& \|x_k-z\|\cdot \left\| \big(J(x_k)\zz J(x_k)\big) \cdot (x_k -  z)  - \big(\nabla f(x_k) - \nabla f(z)\big) \right\| \\
		&=\left\|\int_{S_k}
		\big(J(x_k)\zz J(x_k) \cdot (x_k-z)-\nabla^2 f(x)\cdot (x_k-z)\big) \dd s
		\right\| \\
		&   =\left\|\int_{S_k}
		\big(J(x_k)\zz J(x_k)-J(x)^\top J(x) - H(x) \big)\cdot (x_k-z) \dd s
		\right\| \\
		& \leq
		\left\|\int_{S_k}
		\left(J(x_k)\zz J(x_k) - J(x)\zz J(x)\right)
		\cdot (x_k-z) \dd s\right\|\\
		&
		\quad + \left\|\int_{S_k}
		\big(H(x)-H(z)\big)\cdot (x_k- z) \dd s
		\right\|.
		\end{aligned}
		\end{equation}
	\end{small}
	
	According to Lemma \ref{J(x)} and Corollary \ref{H(x)},  $ J(x)\zz  J(x) $ and $ H(x) $ are Lipschitz continuous on the line
	$ S_k $ with probability at least $ 1-5\exp(-\gamma_\delta n)-4/n^2 $
	provided $ m\geq Cn\log n $. So using (\ref{eq:tuidao}), we obtain
\begin{small}
\begin{equation*}
\begin{aligned}
&\left\| (J(x_k)\zz J(x_k)) \cdot (x_k - z)
- (\nabla f(x_k) - \nabla f(z)) \right\| \\
& \leq \frac{1}{\|x_k-z\|}\cdot\\
&\quad\bigg(\left\|\int_{S_k}
\left(J(x_k)\zz J(x_k) - J(x)\zz J(x)\right)
\cdot (x_k-z) \dd s\right\|\\
&\quad + \left\|\int_{S_k}
		\big(H(x)-H(z)\big)\cdot (x_k- z) \dd s
		\right\|\bigg)\\
		& \leq
		\int_{S_k}
		\left\|J(x_k)\zz J(x_k) - J(x)\zz J(x)\right\|  \dd s\\
		&\quad\quad \quad + \int_{S_k}\|H(x)-H(z)\| \dd s
		\\
		& \leq
		\int_{S_k}
		L_J \|x_k-x\|  \dd s
		+ \int_{S_k}L_H\|x-z\| \dd s
		\\
		& = \frac{L_J+L_H}{2} \cdot \left\|x_k-z\right\|^2\\
		& = 2(7+\frac{3\delta}{4})(1+\sqrt{\delta})\left\|x_k-z\right\|^2.
		\end{aligned}
		\end{equation*}
	\end{small}
	
	Thus according to Lemma \ref{largest_eigen_jj} and (\ref{eq:1}), when $ m\geq Cn\log n $,
	\begin{equation}\label{eq:conc}
	\begin{aligned}
	&\|x_{k+1} -  z \| \\
	&=\|\big(J(x_k)\zz J(x_k)\big)^{-1}\|\cdot \\
	&\quad \quad \|
	\big(J(x_k)\zz J(x_k)\big) \cdot (x_k - z)
	- \big(\nabla f(x_k) - \nabla f(z)\big)
	\|\\
	&\leq\frac{4}{(8-\delta)(1-\sqrt{\delta})^2}\cdot 2(7+\frac{3\delta}{4})(1+\sqrt{\delta})\left\|x_k-z\right\|^2\\
	&= \beta\cdot\left\|x_k-z\right\|^2
	\end{aligned}
	\end{equation}
	holds with probability at least $ 1-c/n^2 $.
	Based on the discussion in Theorem \ref{initial1_1}, we have
	\[
	\|x_{k+1} -  z \|\leq \beta\cdot\left\|x_k-z\right\|^2\leq\sqrt{\delta}.
	\]
	Then we have $ \langle x_{k+1},z \rangle\geq 0 $, i.e., $ \text{dist}(x_{k+1}, z)=\|x_{k+1}-z\| $.
\end{IEEEproof}
\vspace{0.2cm}
{\bf Acknowledgements.} We are  grateful to Xin Liu for discussions and comments at the beginning  of this project, which contributed to the proof of Theorem \ref{th:maintheorem}.
We would like to thank the referees  for thorough and useful comments which
have helped to improve the presentation of the paper.
\vspace{0.2cm}

\bibliographystyle{IEEEtran}
\bibliography{ref}
\end{document}